%% file: Remy.FridovichKeil.ea.ACC.2025.tex
\documentclass[letterpaper, 10 pt, conference]{ieeeconf}  %

\IEEEoverridecommandlockouts                              %

\usepackage[colorlinks]{hyperref}
\usepackage{graphicx}
\usepackage{color}
\usepackage[dvipsnames]{xcolor}
\usepackage{amsmath,amssymb}
\usepackage{mathtools}
\usepackage[skip=2pt]{caption}
\usepackage{epstopdf}
\usepackage{rotating}
\usepackage{cite}
\usepackage{wrapfig}
\usepackage{float}
\usepackage{caption}
\usepackage{subcaption}
\usepackage{tabularx}
\usepackage[compact]{titlesec}

\newtheorem{problem}{Problem}
\newtheorem{example}{Example}
\newtheorem{proposition}{Proposition}
\newtheorem{definition}{Definition}
\newtheorem{remark}{Remark}
\newcommand{\argmin}{\mathop{\mathrm{argmin}}}          %
\usepackage{bbm}
\usepackage{algorithm}
\usepackage{algpseudocode}

\usepackage{xargs} 
\usepackage{xcolor}
\usepackage[colorinlistoftodos,prependcaption,textsize=tiny]{todonotes}

\newcommand{\blfootnote}[1]{
  \begingroup
  \renewcommand\thefootnote{}\footnote{#1}
  \addtocounter{footnote}{-1}
  \endgroup
}

\renewcommand{\baselinestretch}{0.985}	

\input{notation}

\title{\LARGE \bf
Learning responsibility allocations for multi-agent interactions:\\A differentiable optimization approach with control barrier functions
}

\author{Isaac Remy$^{1}$, David Fridovich-Keil$^{2}$, and Karen Leung$^{1,3}$%
\thanks{$^{1}$University of Washington, Department of Aeronautics and Astronautics, $^{2}$University of Texas, Austin, Department of Aerospace Engineering and Engineering Mechanics, $^3$NVIDIA.
        {\tt\small \{iremy,kymleung\}@uw.edu, dfk@utexas.edu}.
        I. Remy was supported by the AI@UW Seed Grant.
        D. Fridovich-Keil was supported by a National Science Foundation CAREER award under Grant No. 2336840.}%
}

\begin{document}

\maketitle
\thispagestyle{empty}
\pagestyle{empty}

\begin{abstract}

From autonomous driving to package delivery, ensuring safe yet efficient multi-agent interaction is challenging as the interaction dynamics are influenced by hard-to-model factors such as social norms and contextual cues. 
Understanding these influences can aid in the design and evaluation of socially-aware autonomous agents whose behaviors are aligned with human values.
In this work, we seek to codify factors governing safe multi-agent interactions via the lens of \textit{responsibility}, i.e., an agent's willingness to deviate from their desired control to accommodate safe interaction with others. 
Specifically, we propose a data-driven modeling approach based on control barrier functions and differentiable optimization that efficiently learns agents' responsibility allocation from data.
We demonstrate on synthetic and real-world datasets that we can obtain an interpretable and quantitative understanding of how much agents adjust their behavior to ensure the safety of others given their current environment. \blfootnote{Project page: \href{https://isremy.github.io/responsibility-website/}{https://isremy.github.io/responsibility-website/}}

\end{abstract}

\section{Introduction}
Safe navigation and collision avoidance with others come naturally to humans, but so much of our decision-making is governed by \textit{social norms}, which depend on hard-to-model factors such as context and interaction history.
For example, how do we precisely quantify factors that influence a driver's decision whether to slow down for an overtaking vehicle?
On the one hand, end-to-end approaches are powerful in capturing complex and nuanced interaction dynamics but lack interpretability in explaining the interactions \cite{NayakantiAlRfouEtAl2023,IvanovicLeungEtAl2020,MizutaLeung2024}. 
On the other hand, handcrafted model-based approaches, while fully interpretable, may miss capturing nuanced corner-cases and subtle interactions \cite{SadighSastryEtAl2016c,MavrogiannisAlves-OliveraEtAl2021}.
Naturally, data-driven model-based methods present an attractive middle-ground for providing interpretability while capturing nuanced interactions.

In this work, we present a data-driven model-based approach to modeling social norms underlying multi-agent interactions. 
Specifically, we model social norms through the lens of \textit{responsibility}, which we define as an agent's willingness to deviate from its desired control to accommodate safe interactions with others (visualized in Fig. \ref{fig:hero}), and present a computationally efficient technique to infer agents' responsibility allocations given information about their surroundings.
Core to our approach is the combination of Control Barrier Functions (CBF) and differentiable optimization which provides interpretable inductive bias and the capability to extract interpretable quantities governing multi-agent interactions from data. Our responsibility learning framework is both applicable to guiding socially-aware robot policy construction and offline evaluation and data analysis.

\noindent\textbf{Contributions.} The contributions of this paper are four-fold: \textbf{(i)} We propose a novel mathematical formalization of describing responsibility allocations for multi-agent interactions, which is based upon Control Barrier Functions.
\textbf{(ii)} We present a computationally efficient technique to learn agents' responsibility allocations from data. 
Computational efficiency is afforded by combining differentiable optimization techniques with modern deep learning and automatic differentiation tools.
\textbf{(iii)} We introduce the concept of \textit{symmetric responsibility} and provide a tractable approach to learning symmetric responsibility allocation models from data. We demonstrate the benefit of symmetric responsibility in improving data efficiency.
\textbf{(iv)} We demonstrate the efficacy of our responsibility learning method on synthetic and human interaction data and highlight its capability to provide interpretable insights into the interaction. 

\begin{figure}
    \centering
    \includegraphics[width=1.0\linewidth]{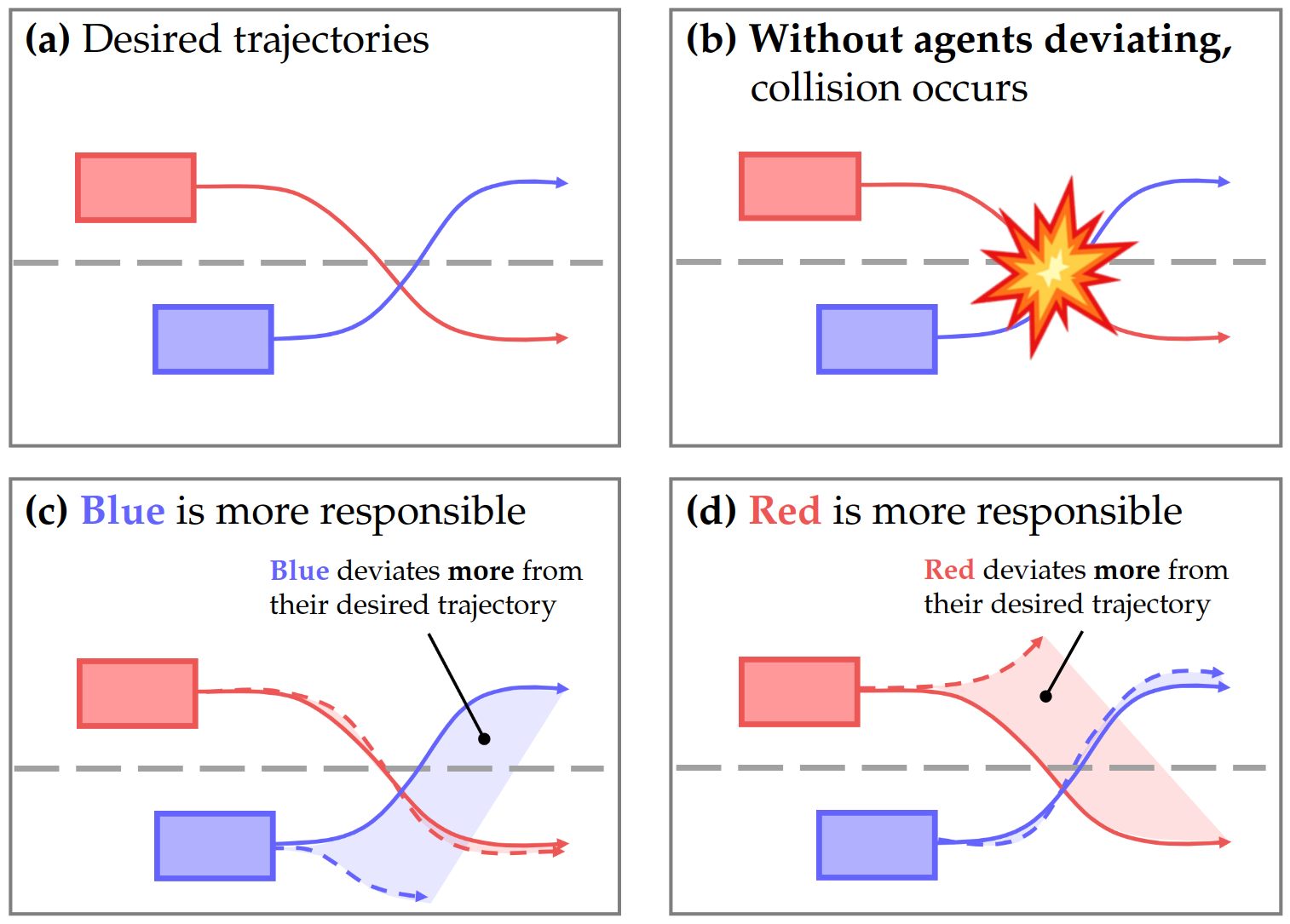}
    \caption{In a) and b), two cars are swapping lanes on a highway, but their desired controls lead to collision. In c) and d), we see how the agents may deviate from their ideal trajectories, according to two different responsibility allocations. In this work we wish to solve the \textit{inverse} problem: how do we infer these allocations from data?}
    \label{fig:hero}
    \vspace{-7mm}
\end{figure}

\section{Related work}

Recent works consider learning and instilling social awareness into autonomous planning algorithms to improve coordination between agents and thus boost the overall safety and efficiency of the multi-agent system \cite{SchwartingPiersonEtAl2019,ToghiValienteEtAl2021,ToghiValienteEtAl2022,SunZhanEtAl2018}. 
The idea is to include another agent's objective function into the autonomous agent's planning objective. The weighting between the two objective terms, referred to as the social value orientation (SVO), determines the output behavior, ranging from prosocial, egotistic, altruistic, competitive, and masochistic.
The SVO weightings can be learned from data \cite{SchwartingPiersonEtAl2019} to understand how much an agent considers the welfare of others.
Although these works do not explicitly discuss responsibility, SVO weightings can potentially encode responsibility as they describe how much one agent considers the welfare of others.
However, a few issues arise. First, the objectives can include other considerations beyond safety, making interpreting the SVO weightings challenging. Second, the SVO paradigm applies to only two agents, and it is unclear how it can extend to multiple agents. Third, it relies on knowledge of other agents' reward functions, which is often hard to obtain, especially for multi-agent settings.

There are a few recent works that explicitly consider the concept of responsibility in terms of multiple agents satisfying a shared collision avoidance constraint.
Optimal reciprocal collision avoidance (ORCA) is a popular collision avoidance strategy, and the original concept assumes agents share equal responsibility in avoiding the velocity obstacle\cite{VandenBergLinEtAl2008}.
Various works have changed the responsibility allocation, either via heuristics or based on agents' states \cite{YinLiuEtAl2019,LuoCaiEtAl2018,LuoCaiEtAl2022}. 
More recently, \cite{GuoWangEtAl2021} proposes a variable responsibility optimal reciprocal collision avoidance (VR-ORCA) which computes a responsibility allocation that minimizes each agent's deviation from their preferred velocity. This choice of responsibility allocation results in what is more convenient for agents rather than based on what social norms would dictate. Although the concept of responsibility is explicitly modeled, the proposed formulation is not compatible with learning from data to mimic how humans naturally interact with one another. 
Lastly, similar to our proposed approach, \cite{CosnerChenEtAl2023} also frames responsibility allocations using CBFs but instead views responsibility as a parameter denoting how much each agent contributes to a joint CBF control constraint. In contrast, our responsibility parameter is a value describing how much agents deviate from their desired control---interesting future work could lie in more rigorous comparison of different responsibility formulations.

\section{Problem formulation}
Consider a system with $N$ agents where $x\in\mathcal{X}$ denotes the state describing all agents, $u = (u_1,u_2,...,u_N)\in\mathcal{U}_1\times\mathcal{U}_2\times...\times\mathcal{U}_N$ denote the control inputs of each agent, and $\dot{x} = f(x, u)$ denotes the dynamics of the multi-agent system. 
Let $g(x, u)\geq 0$ denote a collision avoidance constraint shared amonst all agents.
In real-world settings, humans take control actions to accomplish their desired task and avoid collision with one another.
But \textit{how much does each agent contribute to ensuring the collision avoidance constraint is satisfied?}
In other words, \textit{how responsible} is each agent in taking control actions that ensures that $g(x,u)\geq 0$ is satisfied.
Suppose $\pi_{1:N}:\mathcal{X} \rightarrow \mathcal{U}_1 \times \mathcal{U}_2\times...\times\mathcal{U}_N$ is an arbitrary policy describing how $N$ humans behave around each other.
\textit{How does $\pi$ encode each agent's responsibility allocation in satisfying collision avoidance constraints?}

While a human's true behavior policy is extremely complex, we seek to develop an interpretable lens on $\pi$ to describe each agent's responsibility allocation. 
Let $\gammavec \in \mathbb{R}^N$ denote the responsibility allocation for each agent. Then we seek to answer the following two questions: 
\textbf{(i) \textit{What is an interpretable model describing the dependency of $\pi$ on $\gammavec$?}
(ii) \textit{Given such a model, how can $\gammavec$ be tractably inferred from human interaction data?}}

\section{Defining responsibility allocations}
In this section, we define responsibility as an agent's inclination to deviate from their desired control to satisfy collision avoidance constraints.
At a high level, suppose agents have a desired control action they would like to take, e.g., maintain speed in their lane, or change lanes. Then, a \textit{more (less) responsible} agent is willing to deviate more (less) from their desired control to respect collision avoidance constraints. 
We model agents' willingness to deviate from their desired control action as an optimization problem based on a safety filter derived from a control barrier function (CBF).  Our use of CBFs specifically as a model for human safety is motivated by their simple but descriptive nature---CBFs let us describe human collision avoidance constraints concisely, intuitively, and with relatively minor assumptions.
We first introduce CBFs and describe their use as a safety filter, and then formally define responsibility allocations.

\subsection{Control barrier function (CBFs)}
Control Barrier Functions (CBFs) are mathematical tools used in control theory to ensure the safety of dynamic systems while achieving desired control objectives. 
Consider continuous-time dynamics, $\dot{x}=f(x,u)$ where $x\in\mathcal{X}\subset \mathbb{R}^n, \: u\in \mathcal{U} \subset \mathbb{R}^m$, and $f: \mathcal{X}\times \mathcal{U} \rightarrow \mathbb{R}^n$ is locally Lipschitz continuous. A CBF is defined as follows.

\begin{definition}[Control Barrier Function (CBF) \cite{AmesGrizzleEtAl2014}]
Given the aforementioned system dynamics, consider a continuously differentiable function $b: \mathcal{X} \rightarrow \mathbb{R}$ and a set $\mathcal{C} \subset \mathcal{X}$ where $\mathcal{C} = \{ x \in \mathcal{X} \mid b(x) \geq 0\}$. Then $b$ is a valid  CBF if there exists a class $\mathcal{K}_\infty$ function $\alpha$ such that,
\begin{equation}
    \sup_{u\in \mathcal{U}}  \: \nabla b(x)^Tf(x,u) + \alpha(b(x)) \geq 0,\quad  \forall x\in\mathcal{X}.
    \label{eq:cbf}
    \vspace{-5mm}
\end{equation}
\label{def:cbf}
\end{definition}

Intuitively, $b(\cdot)$ is a measure of safety, and a valid CBF ensures that there exists a feasible control preventing $b(x(t))$ from decreasing faster than a rate of $-\alpha(b(x))$ along system trajectories and that $\dot{b}(x(t))=0$ along the boundary of $\mathcal{C}$.
It follows that if $b(\cdot)$ is a valid CBF for a set $\mathcal{C}$, then any locally Lipschitz controller $k: \mathcal{X} \rightarrow \mathcal{U}$ satisfying $\nabla b(x)^Tf(x,u) + \alpha(b(x)) \geq 0$ renders the set $\mathcal{C}$ forward invariant \cite{AmesXuEtAl2017}.

\subsection{CBF safety filters}
Given the forward invariance property of CBFs, we can compute a set of control inputs that prevent the system from exiting $\mathcal{C}$.
The control set is simply the set of controls that satisfy the CBF inequality in \eqref{eq:cbf}, i.e.,
\vspace{-3mm}

{\small
\begin{equation}
\mathcal{U}^\safe(x; b, \alpha) = \lbrace u\in\mathcal{U} \mid \nabla b(x)^Tf(x,u) + \alpha(b(x)) \geq 0 \rbrace.
\label{eq:safe control set}
\end{equation}
}

\begin{problem}[CBF safety filter]
A safety filter will take the control from a \textit{desired} policy, $u^\desire = \pi^\desire(x)$, and project it into the safe control set $\mathcal{U}^\safe(x; b, \alpha)$ defined in \eqref{eq:safe control set}, i.e.
{\small
\begin{align*}
    u^\star = \argmin_{u\in\mathcal{U}} \: \| u - u^\desire\|_2^2\:\: \text{s.t.}  \:\: \nabla b(x)^Tf(x,u) + \alpha(b(x)) \geq 0.
\end{align*}
}
\label{prob:CBF safety filter}
\vspace{-5mm}
\end{problem}

\begin{remark}
If the dynamics are control affine, i.e., $\dot{x} = f(x) + g(x)u$ and $u\in\mathcal{U}$ are box constraints, i.e. $u_{\min} \leq u \leq u_{\max}$, then Prob.~\ref{prob:CBF safety filter} is a convex quadratic program (QP) and can be solved efficiently \cite{StellatoBanjacEtAl2017, TracyManchester2024}.
\label{re:QP}
\end{remark}

\begin{remark}
\label{rem:slack}
    For a valid CBF, Prob.~\ref{prob:CBF safety filter} is always feasible. However, finding a valid CBF is often nontrivial. Instead, a desired CBF can be chosen and slack variables can be introduced (demonstrated in Prob. \ref{prob:vanilla multiagent CBF safety filter}) to encourage agents to satisfy the CBF constraint as best as possible within the control limits. In the context of this work, the loss of guarantees is not an issue as we are analyzing how much deviation each agent makes from their desired control rather than directly synthesizing guaranteed safe control strategies.
\end{remark}

\begin{remark}
    Definition~\ref{def:cbf} assumes we have a relative degree 1 system. To account for higher relative degree systems (e.g., double integrator systems, with $b(x)$ a function only of position variables), we can employ high order CBFs \cite{XiaoBelta2021} which results in a slightly more complicated CBF inequality that remains linear in control for control affine systems.
    \label{re:high order cbf}
\end{remark}

\subsection{Responsibility allocations via a CBF safety filter lens}

In this work, we use CBF safety filters as models to describe \textit{multi-agent} interactions, allowing us to infer how agents cooperate to maintain safety.
The insights gained into multi-agent interactions can be used for offline (crash) analysis, measuring social acceptability of a policy, or to adapt an agent's behavior online.

Consider a \textit{multi-agent} control affine dynamical system where $\mathbf{x}=[x_1,...,x_N]$ is the joint state describing $N$ agents, and each agent has control input $u_i$.
\begin{equation}
    \dot{\mathbf{x}} = \tilde{f}(\mathbf{x}) + \sum_{i=1}^N g_i(\mathbf{x})u_i
    \label{eq:multiagent dynamical system}
\end{equation}
Given a CBF $b(\cdot)$ for the multiagent system, Prob.~\ref{prob:vanilla multiagent CBF safety filter} describes the direct application of the CBF safety filter problem (Prob.~\ref{prob:CBF safety filter}) onto the joint dynamics under $u^\desire_{1:N} = \pi^\desire_{1:N}(\mathbf{x})$.

\begin{problem}[Unweighted multiagent CBF safety filter]
{\small
\begin{align*}
    u^\execute_{1:N} = &\argmin_{u_1,...,u_N} \quad \sum_{i=1}^N\| u_i - u^\desire_i\|_2^2\\
     &\:\:\mathrm{s.t.}  \:\: \nabla b(\mathbf{x})^T\left[ \tilde{f}(\mathbf{x}) + \sum_{i=1}^N g_i(\mathbf{x})u_i\right] + \alpha(b(\mathbf{x})) \geq 0\\
      & \:\quad\quad u_1\in\mathcal{U}_1,..., u_N\in\mathcal{U}_N.
\end{align*}
}
\vspace{-7mm}
\label{prob:vanilla multiagent CBF safety filter}
\end{problem}

Notice that Prob.~\ref{prob:vanilla multiagent CBF safety filter} has \textit{equal} weighting in penalizing each agent for deviating from their desired control input, implying no agent is more or less inclined to ensure the safety constraint is satisfied.
We modify Prob.~\ref{prob:vanilla multiagent CBF safety filter} by adding a coefficient $\gammavec_i$ for each agent to influence the relative penalty in deviating away from the desired control.
As such, we propose the following definition of responsibility allocation.

\begin{definition}[Responsibility allocation]
Let there be $N$ agents interacting, and their joint dynamics are governed by \eqref{eq:multiagent dynamical system}. Let $b(\cdot)$ be a CBF for the joint system describing the collision set. Each agent $i \in \{1, 2, \dots, N\}$ has a desired control input $u_i^\desire$ that is projected into the safe control set via a projection mapping $\mathrm{proj}$.
The projection mapping $\mathrm{proj}$ depends on the corresponding state $\mathbf{x}$, CBF $b$, class $\mathcal{K}_\infty$ function $\alpha$, and the responsibility allocation vector $\gammavec=[\gammavec_1, \gammavec_2,..., \gammavec_N]$ with $\mathbf{1}^T\gammavec = 1$ and $0 \leq \gammavec_i \leq 1$, which determines how much deviation from the desired control each agent is willing to make to satisfy 
\eqref{eq:safe control set}.    
The projection mapping $\mathrm{proj}$ can be computed by solving the optimization problem described in Prob.~\ref{prob:responsible multiagent CBF safety filter} (differences from Prob.~\ref{prob:vanilla multiagent CBF safety filter} are highlighted in blue).
\vspace{1mm}
\begin{problem}[Responsible multiagent CBF safety filter]
{\small
\begin{align*}
    \mathrm{proj}&(u_{1:N}^\desire; \mathbf{x}, b, \alpha, \gammavec) \coloneqq  \\
    &\argmin_{u_1,...,u_N, {\color{blue}\epsilon}}  \: \sum_{i=1}^N   \left({\color{blue}\gamma_i} \| u_i - u^\desire_i\|_2^2 + {\color{blue}\beta_1 \|u_i\|_2^2} \right) + {\color{blue}\beta_2 \epsilon^2}\\
     &\mathrm{s.t.}   \quad \nabla b(\mathbf{x})^T  \left[  \tilde{f}(\mathbf{x}) + \sum_{i=1}^N g_i(\mathbf{x})u_i\right] + \alpha(b(\mathbf{x})) \geq{\color{blue}-\epsilon}\\
      &\qquad\:\: u_1\in\mathcal{U}_1,..., u_N\in\mathcal{U}_N\\
      & \qquad\:\: {\color{blue}\epsilon \geq 0}.
\end{align*}
}
\label{prob:responsible multiagent CBF safety filter}
\vspace{-5mm}
\end{problem}
\end{definition}

\begin{remark}
    The slack variable $\epsilon >0$ reflects the discussion in Remark \ref{rem:slack}. The regularization term $\|u_i\|_2^2$ ensures a unique solution for cases in which Agent $i$ bears total responsibility for constraint satisfaction, i.e. $\gammavec_i=0$. $\beta_{1,2}$ are scalar hyperparameters for regularization and $\epsilon$ weighting.
\end{remark}

\begin{example}[Two-agent 1D single integrator]
    \label{eg: 2-agent 1D single integrator}
    Consider a two-agent 1D single integrator system with CBF,
    \begin{align*}
        \begin{bmatrix}
            \dot{x}_1\\ \dot{x}_2
        \end{bmatrix} = \begin{bmatrix}
            u_1\\ u_2
        \end{bmatrix},\: b(x_1, x_2) = (x_1 - x_2)^2 - 1,\: \alpha(x) = x.
    \end{align*}
    Let $\gamma \in [0, 1]$ denote the responsibility allocation of Agent 1. 
    Fig.~\ref{fig:cbf filter} shows the solution to Prob.~\ref{prob:responsible multiagent CBF safety filter} for various values of $\gamma$.
    We see that when $\gamma=0$, Agent 1 takes full responsibility by deviating from its desired control value to satisfy the CBF constraint (green region) while Agent 2 has no deviation. As $\gamma$ increases, Agent 2 begins to deviate more and Agent 1 less, until at $\gamma=1$, Agent 2 is fully responsible.

    \begin{figure}[h]
        \centering
        \vspace{-3mm}
        \includegraphics[width=0.6\linewidth]{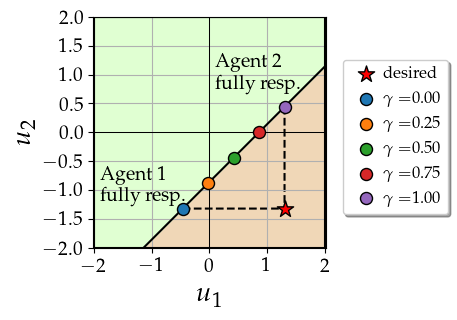}
        \caption{Solutions to the two-agent CBF filter problem with different responsibility allocation values for Agent 1. The green and orange regions are the sets of feasible and infeasible controls, respectively.}
        \label{fig:cbf filter}
        \vspace{-5mm}
    \end{figure}
\end{example}

\section{Inferring responsibility from data}
The previous section described how a choice of responsibility allocation $\gammavec$ influences each agent's deviation from their desired control input. But \textit{how do we choose responsibility allocation values?} 
Human expertise can qualitatively describe trends (e.g., the agent behind should yield to the agent in front), 
but data quantitatively informs what the responsibility allocation values should be.
To this end, we propose a data-driven inverse optimization approach to learn responsibility allocations from interaction data.

\subsection{Bi-level optimization formulation}
The goal is to find a value for $\gammavec$ such that the solution to Prob.~\ref{prob:responsible multiagent CBF safety filter} matches an interaction dataset as closely as possible.
Suppose we have a dataset describing $N$-agent interactions $\mathcal{D} = \lbrace (\mathbf{x}^{(k)}, u_{1:N}^{(k)}) \rbrace_{k=1,...,K}$, and a desired policy $\pi^\desire_{1:N}(\mathbf{x})$ producing the agents' desired control inputs $u^\desire_{1:N}$. For example, each agent may wish to ignore the others and follow a nominal trajectory.
Then, for a choice of CBF parameters $b$ and $\alpha$, we want a responsibility allocation vector $\gammavec$ such that $u^\data_{1:N} \approx \mathrm{proj}(\pi_{1:N}^\desire(\mathbf{x}^\data); \mathbf{x}^\data, b, \alpha, \gammavec)$ for some distance metric $\Delta$ (e.g., L1, L2, Huber loss).
Since the projection mapping $\mathrm{proj}$ is an optimization problem (see Prob.~\ref{prob:responsible multiagent CBF safety filter}), the resulting regression problem has a bi-level structure.

\begin{problem}[Responsibility allocation inference]

\begin{align*}
\min_{\gammavec} &\: \frac{1}{|\mathcal{D}|}\sum_{(\mathbf{x}^i, u_{1:N}^i) \in \mathcal{D}} \Delta( u_{1:N}^i, \tilde{u}_{1:N}^i(\gammavec))\\
\mathrm{s.t.} & \: \tilde{u}_{1:N}^i(\gammavec) = \underbrace{\mathrm{proj}(\pi_{1:N}^\desire(\mathbf{x}^i); \mathbf{x}^i, b, \alpha, \gammavec),}_{\text{See Prob.~\ref{prob:responsible multiagent CBF safety filter}}} \\
&\quad \mathbf{1}^T\gammavec = 1, \quad 0 \leq \gamma_i \leq 1.
\end{align*}
\label{prob:responsibility bilevel}
\vspace{-7mm}
\end{problem}
Prob.~\ref{prob:responsibility bilevel} consists of \textit{multiple} inner problems, one for each data point. 
This problem structure promotes the use of parallelization as an efficient solution method.
Next, we describe how we can efficiently solve Prob.~\ref{prob:responsibility bilevel}.

\subsection{Leveraging differentiable optimization}
\label{subsec:solving diff opt}

We exploit recent results in differentiable optimization \cite{AmosKolter2017,AgrawalAmosEtAl2019,AgrawalBarrattEtAl2020,AgrawalBarrattEtAl2021,TracyManchester2024} and automatic differentiation tools \cite{JAX2018,PaszkeGrossEtAl2017,RevelsLubinEtAl2016} to tractably solve Prob.~\ref{prob:responsibility bilevel} using gradient descent.
We first present the gradient descent algorithm and then discuss how this can be solved efficiently.

The procedure for taking a single gradient descent step is described in Alg.~\ref{alg:gradient descent}.
Line \ref{algline:gamma constraint} ensures that the constraints on $\gammavec$ are met rather than employing a projected gradient descent.
Line \ref{algline:vmap} describes solving many instances of Prob.~\ref{prob:responsible multiagent CBF safety filter} (i.e., batching the computation) and collating the solutions as $\tilde{U}_{1:N}$ which is used in computing the loss in line \ref{algline:loss}. 
Line \ref{algline:autodiff} performs the update step with step size $\delta$. The gradient $\frac{\partial \ell}{\partial \gammavec} = \frac{\partial \ell}{\partial \gammavec} \frac{\partial \gammavec}{\partial\widetilde{\gammavec}}$ requires differentiating through \textit{all} the instances of Prob.~\ref{prob:responsible multiagent CBF safety filter} with respect to $\gammavec$.

\begin{algorithm}[t]
\caption{Single gradient step to solve responsibility allocation inference}\label{alg:gradient descent}
\begin{algorithmic}[1]
\Require $\widetilde{\gammavec}$ (unnormalized responsibility allocation)
\Require $\delta$ (step size)
\Require $\mathcal{D} = \lbrace (\mathbf{x}^{(k)}, u_{1:N}^{(k)}) \rbrace_{k=1,...,K}$ (dataset)
\State $\gammavec = \mathrm{softmax}(\widetilde{\gammavec})$ (ensuring $\mathbf{1}^T\gammavec=1$ and $\gammavec\in[0,1]$) \label{algline:gamma constraint}
\State $\widetilde{U}_{1:N} = \mathrm{batch}(\tilde{u}_{1:N}(\gammavec, \mathbf{x}^\data); \mathcal{D})$ \label{algline:vmap}
\State $\ell(\gammavec) = \sum_{(u_{1:N}^\data, \tilde{u}_{1:N}) \in \mathcal{D}\times \widetilde{U}_{1:N}} \Delta( u_{1:N}^\data, \tilde{u}_{1:N})$ \label{algline:loss}
\State $\widetilde{\gammavec} \gets \widetilde{\gammavec} - \delta \frac{\partial \ell}{\partial \gammavec} \frac{\partial \gammavec}{\partial\widetilde{\gammavec}}$ \label{algline:autodiff}
\end{algorithmic}
\end{algorithm}
Referring back to Remark~\ref{re:QP}, if the dynamics are control affine, then  Prob.~\ref{prob:responsible multiagent CBF safety filter} is a quadratic program.
This means that (i) the problem is convex and can be efficiently solved, and (ii) we can differentiate through convex programs \cite{AmosKolter2017,AgrawalAmosEtAl2019,AgrawalBarrattEtAl2020,AgrawalBarrattEtAl2021,TracyManchester2024}, 
and (iii) with the recent advancements in automatic differentiation tooling and batching, such as JAX \cite{JAX2018}, we can efficiently solve many instances of a QP and differentiate through them in parallel \cite{TracyManchester2024}.

\subsection{Context-dependent responsibility}
Sec.~\ref{subsec:solving diff opt} described the algorithm for finding the responsibility allocation parameter $\gammavec$ using gradient descent. That formulation assumed $\gammavec$ is a fixed vector describing agent responsibilities for an entire dataset.
However, in general, the responsibility allocation is not agent-specific, but rather state- or context-dependent (i.e., $\gammavec$ is a function). 
Given that we are solving Prob.~\ref{prob:responsibility bilevel} using gradient descent, Alg.~\ref{alg:gradient descent} can easily extend to the case where $\gammavec=\mathrm{softmax}(h_\theta(\mathbf{x}, e))$ where $h_\theta:\mathcal{X}^N \times \mathcal{E} \rightarrow[0,1]^N$ is a differentiable function parameterized by $\theta$, such as a neural network, and depends on the state $\mathbf{x}$ and environment information $e$, such as map/road information. The $\mathrm{softmax}$ function ensures the $\gammavec$ constraints are satisfied.
As such, we can perform gradient descent on the parameters $\theta$ instead.

\begin{figure*}[t]
\begin{subfigure}{.32\textwidth}
  \centering
  \includegraphics[width=\linewidth]{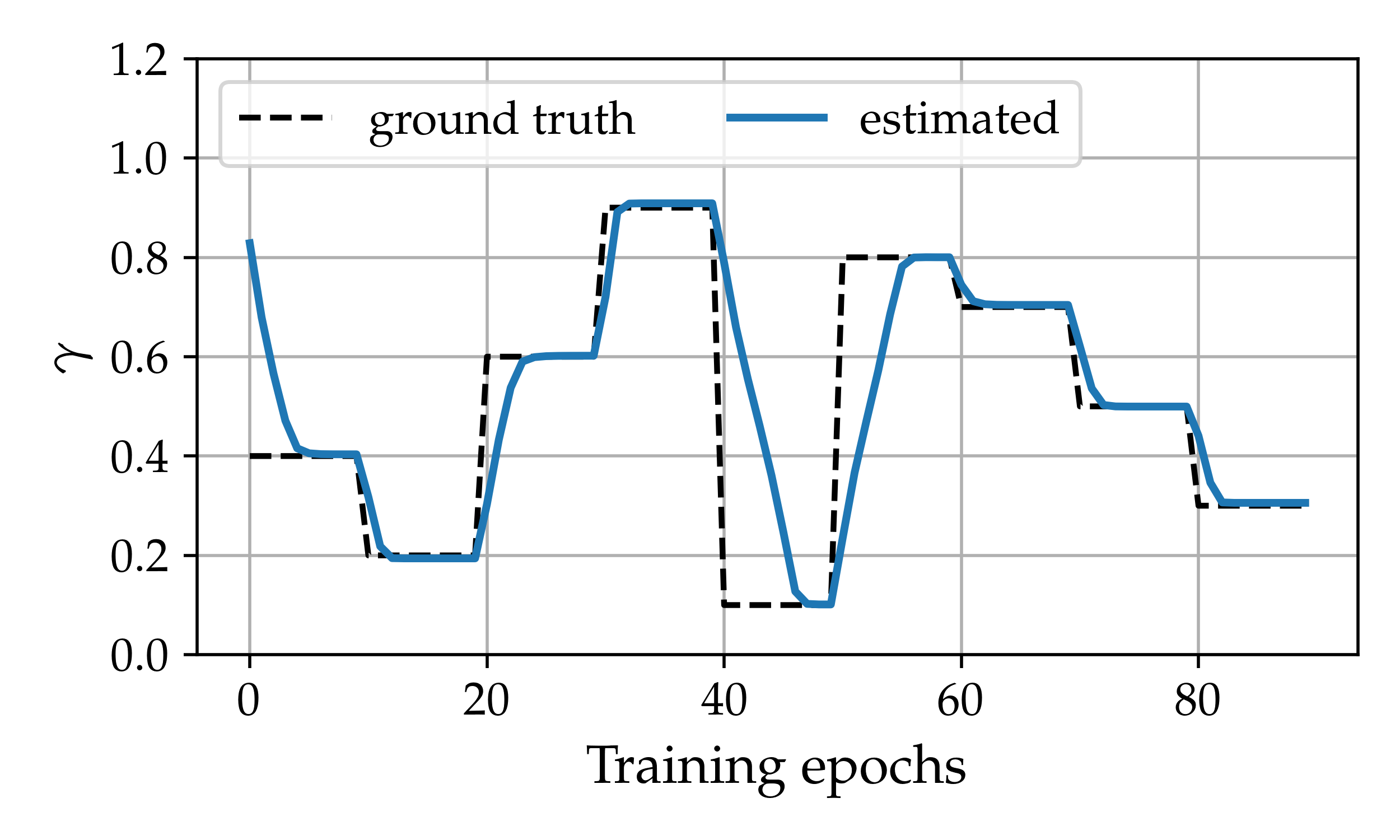}
  \vspace{-7mm}
  \caption{Time-varying $\gamma$ for a 2-agent system.}
  \label{fig:2-agent}
\end{subfigure}%
\begin{subfigure}{.32\textwidth}
  \centering
  \includegraphics[width=\linewidth]{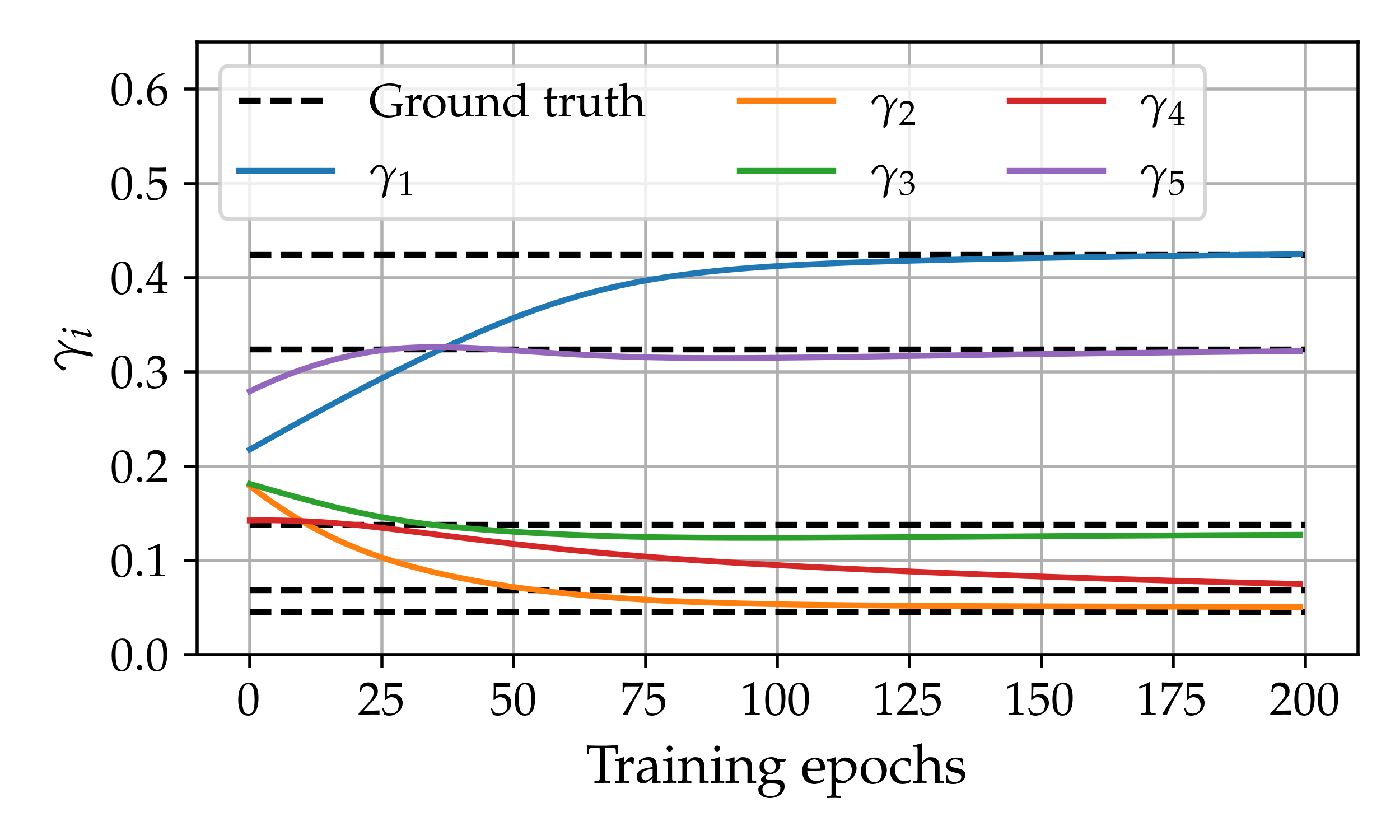}
  \vspace{-7mm}
  \caption{$\gamma$ for each agent in a 6-agent system.}
  \label{fig:6-agent}
\end{subfigure}
\begin{subfigure}{.32\textwidth}
  \centering
  \includegraphics[width=\linewidth]{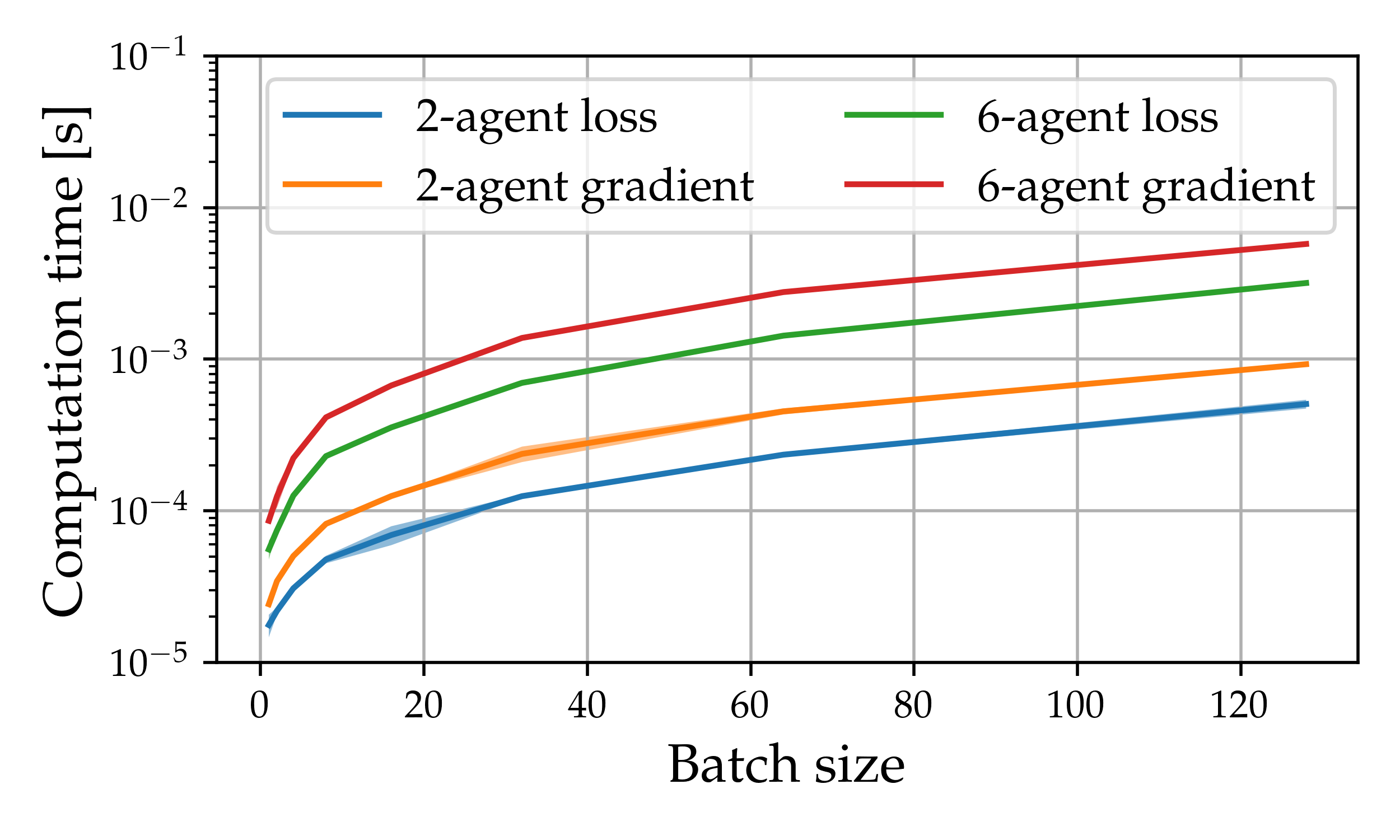}
  \vspace{-7mm}
  \caption{Computation time as a function of batch size.}
  \label{fig:synthetic computation time}
\end{subfigure}
\caption{We generate ground truth $\gamma$ values to create synthetic $u_{1:N}$ via Prob. \ref{prob:responsible multiagent CBF safety filter}, then learn the original $\gamma$ values via Prob. \ref{prob:responsibility bilevel}. After a few epochs, the estimated responsibility allocation value $\gamma$ converges to the ground truth. Computation time scales linearly with batch size.}
\label{fig:synthetic data results}
\vspace{-5mm}
\end{figure*}

\subsection{Enforcing symmetry in responsibility allocation}
\label{subsec:symmetry}

In learning state and/or context-dependent responsibility allocations, there may be some symmetries we would like to exploit. 
In particular, the assigned number ordering of agents (i.e., who is labeled Agent 1, Agent 2,..., Agent $N$) should not matter. 
We introduce a \textit{symmetry responsibility}.
\begin{definition}[Symmetric responsibility]
    For an $N$ agent system with joint state space $\mathcal{X}^N$, suppose some assignment ordering $(1, 2, ..., N)$. Let $\gammavec: \mathcal{X}^N\rightarrow[0,1]^N$ be the responsibility allocation function, and $\gamma_i(\mathbf{x})$ (with $\mathbf{x}=[x_1,...,x_N]$) denote the responsibility allocation for agent $i$. Then, we say the system has symmetric responsibility allocation if the following properties hold:
    \begin{align}
    &\gamma_i(\mathbf{x}) = \gamma_i(\sigma(\mathbf{x})) \: \forall \sigma \in \mathcal{P}_{i}, \: \forall i = 1,...,N \label{eq:symmetry}\\
    & \gamma_i(\mathbf{x}) = \gamma_j(\mathbf{x}_{i\leftrightarrow j})  \: \forall i,j=1,...,N\label{eq:invariance}\\
    &\sum_{i=1}^N \gamma_i(\mathbf{x}) = 1 \label{eq:complementary} 
    \end{align}
    where $\mathcal{P}_{i}$ is the set of all permutations (for $N$ agents) that keep index $i$ fixed, and $\mathbf{x}_{i\leftrightarrow j}$ denotes that the entries at indices $i$ and $j$ are swapped.
\end{definition}
The first equation \eqref{eq:symmetry} describes that the responsibility allocation for agent $i$ is invariant when the assignment ordering of other agents were permuted.
The second equation \eqref{eq:invariance} describes that for a given agent, their responsibility allocation should remain the same regardless of their assigned numbering.
The third property \eqref{eq:complementary} ensures that the sum of all agents' responsibilities is equal to one.
We can construct a function satisfying the symmetric responsibility property.

First, let us define $\sigma:\mathcal{X}^N\rightarrow \mathcal{X}^N$ as a permutation function that permutes the input ordering. We further write $\sigma_{i\rightarrow j}$ to denote any permutation that moves index $i$ to $j$.

\begin{proposition}
Let $\widetilde{\phi}:\mathcal{X}^N \rightarrow \mathbb{R}$ be any (differentiable) function, e.g., a neural network. Then we define $\gamma_i$,

{\small
\begin{align}
    \gamma_i(\mathbf{x}) = \frac{\exp(\phi(\sigma_{i\rightarrow 1}(\mathbf{x})))}{\sum_{j=1}^N \exp(\phi(\sigma_{j\rightarrow 1}(\mathbf{x})))}, \quad \phi(\mathbf{x}) = \sum_{\sigma\in\mathcal{P}_{1}} \widetilde{\phi}(\sigma(\mathbf{x}))
    \label{eq:gamma}
\end{align}
}

Then $\gammavec(\mathbf{x}) = [\gamma_1(\mathbf{x}),...,\gamma_N(\mathbf{x})]$ satisfies the symmetric responsibility properties.
\end{proposition}

\begin{proof}
First, we prove that $\gamma_i$ satisfies \eqref{eq:symmetry} for all $i$.
By construction, $\exp(\phi(\sigma_{i\rightarrow 1}(\mathbf{x})))$ is constant under permutations from $\mathcal{P}_i$.
We need to show that the denominator of $\gamma_i$ is invariant under any permutation.
For any permutation $\sigma$, denote $\zeta(j)$ as the index that maps to $j$.
We write $\mathbf{x}_{k\rightarrow l}$ to highlight that index $k$ has been moved to index $l$, with other indices potentially also shuffled.
For any $\sigma \in \mathcal{P}$, we have,

\vspace{-3mm}{\small
\begin{align*}
    \sum_{j=1}^N \exp(\phi(\sigma_{j \rightarrow 1}(\sigma(\mathbf{x})))) &= \sum_{j=1}^N \exp(\phi(\sigma_{j \rightarrow 1}(\mathbf{x}_{\zeta(j) \rightarrow j}))) \\
    &= \sum_{j=1}^N \exp(\phi(\sigma_{\zeta(j) \rightarrow 1}(\mathbf{x}))) \\
    &= \underbrace{\sum_{k=1}^N \exp(\phi(\sigma_{k \rightarrow 1}(\mathbf{x})))}_{\text{Denominator of }\gamma_i} 
\end{align*}
}
Leveraging the fact that the denominator is invariant under permutations, property \eqref{eq:invariance} follows,
\begin{align*}
    \gamma_j(\mathbf{x}_{i\leftrightarrow j}) &= \frac{\exp(\phi(\sigma_{j\rightarrow 1}(\mathbf{x}_{i\leftrightarrow j})))}{\sum_{k=1}^N \exp(\phi(\sigma_{k\rightarrow 1}(\mathbf{x}_{i\leftrightarrow j})))}\\
    &= \frac{\exp(\phi(\sigma_{i\rightarrow 1}(\mathbf{x})))}{\sum_{k=1}^N \exp(\phi(\sigma_{k\rightarrow 1}(\mathbf{x})))}\\
    &= \gamma_i(\mathbf{x}).
\end{align*}

Property \eqref{eq:complementary} is met via use of the $\mathrm{softmax}$ function.
\end{proof}

Constructing a symmetric responsibility allocation function amounts to defining a single function $\widetilde{\phi}:\mathcal{X}^N \rightarrow \mathbb{R}$ and applying all possible $N-1$ permutations on it. As such, this construction can become prohibitively expensive as $N$ increases. In this work, we study cases with small $N$ values for which the computation is tractable, and defer investigating scalable approaches for large $N$ in future work.

\begin{figure}[t]
    \centering
    \includegraphics[width=1.0\linewidth]{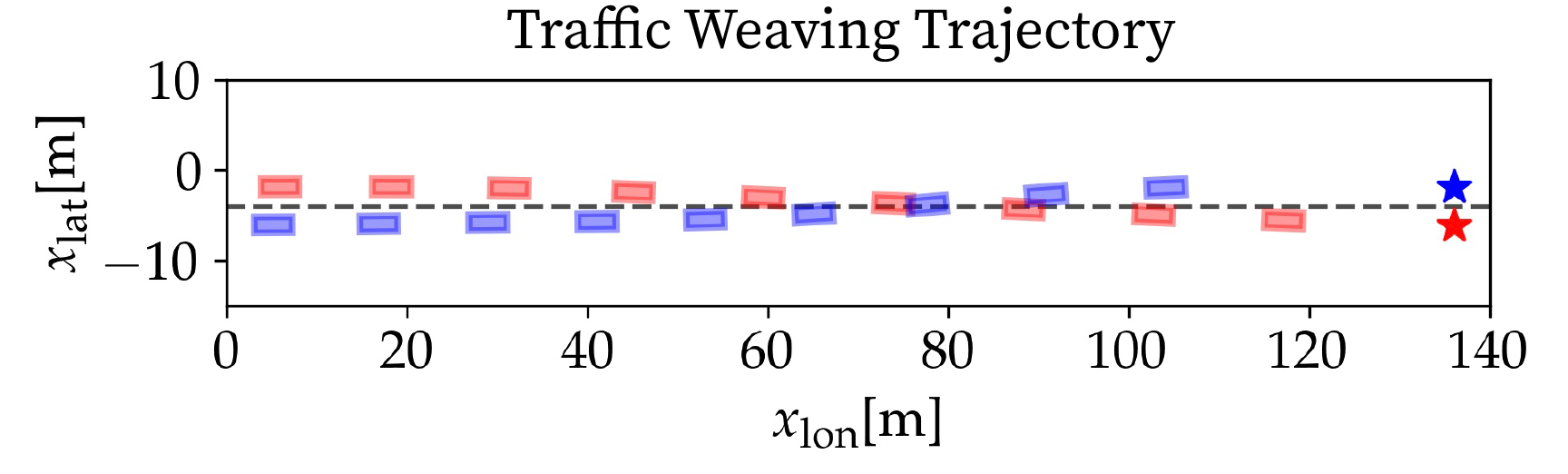}
    \caption{An example trajectory from the traffic weaving dataset. Two cars (red and blue rectangles) start in adjacent lanes, and need to change lanes safely to reach their lateral goals (stars).}
    \label{fig:full-traj}
    \vspace{-5mm}
\end{figure}

\subsection{Symmetry in relative coordinates for two-agent systems}
In the case of a two-agent setting, it may be more convenient to express the system in terms of \textit{relative dynamics} rather than concatenating the states. 
Let $\mathbf{r} = (r_1,...,r_n) \in \mathcal{R}$ denote the relative coordinate.
The state space formulation from Sec.~\ref{subsec:symmetry} must be transferred to the relative coordinate frame.
Swapping the assignment of agents is equivalent to negating the relative coordinates. 
Thus we have $\gamma_1(\mathbf{r}) + \gamma_1(-\mathbf{r}) = 1$. Essentially, rather than requiring invariance under permutations, we require invariance under negation.

We can similarly construct a function satisfying this complementary symmetry for relative coordinates.
Let $\phi: \mathcal{R} \rightarrow \mathbb{R}$ be any function, e.g., neural network. Then,
\begin{align}
    \gamma_1(\mathbf{r}) = \frac{1}{2}(1 + \tanh( \phi(\mathbf{r}) - \phi(-\mathbf{r}))
    \label{eq:relative gamma}
\end{align}
satisfies the analogous symmetry conditions.

\begin{figure*}[t]
    \centering
    \includegraphics[width=\textwidth]{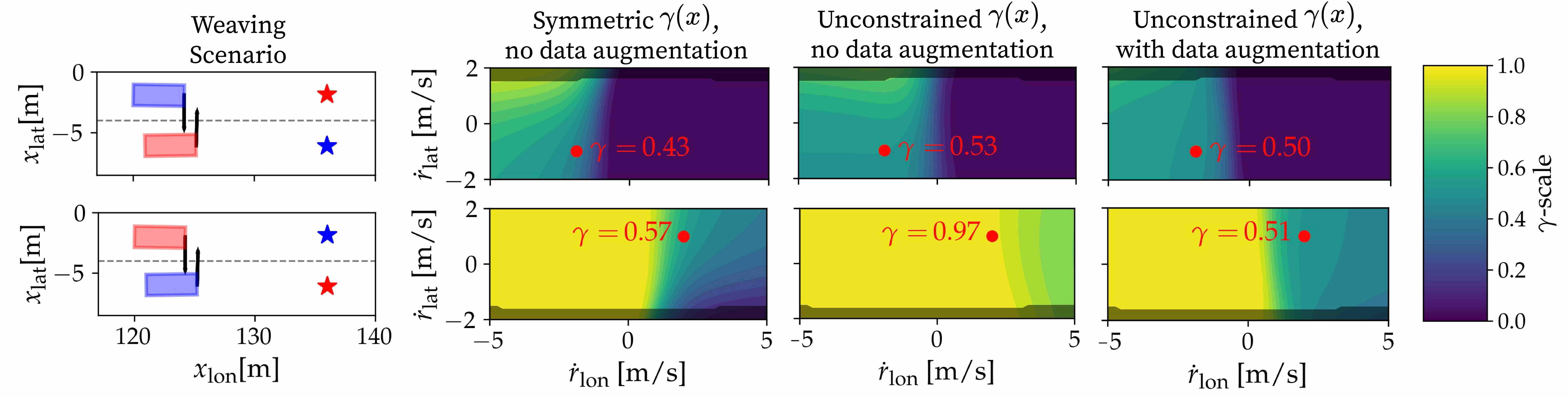}
    \caption{Comparison of $\gamma$ landscapes for the red car using symmetric and unconstrained models. In the scenario, both cars want to change lanes (stars denote each agent's desired goal). The car in the lower lane is moving faster. The bottom row is the same scenario but with the agents swapped. The black arrows in the left plot represent each agent's desired control vector, and the red point in the contour plots corresponds to the depicted scenarios. The shaded regions denote relative states for which the CBF filter is not active. 
    }
    \vspace{-4mm}
    \label{fig:symmetry-contours}
\end{figure*}

\begin{figure}
    \centering
    \includegraphics[width=1.0\linewidth]{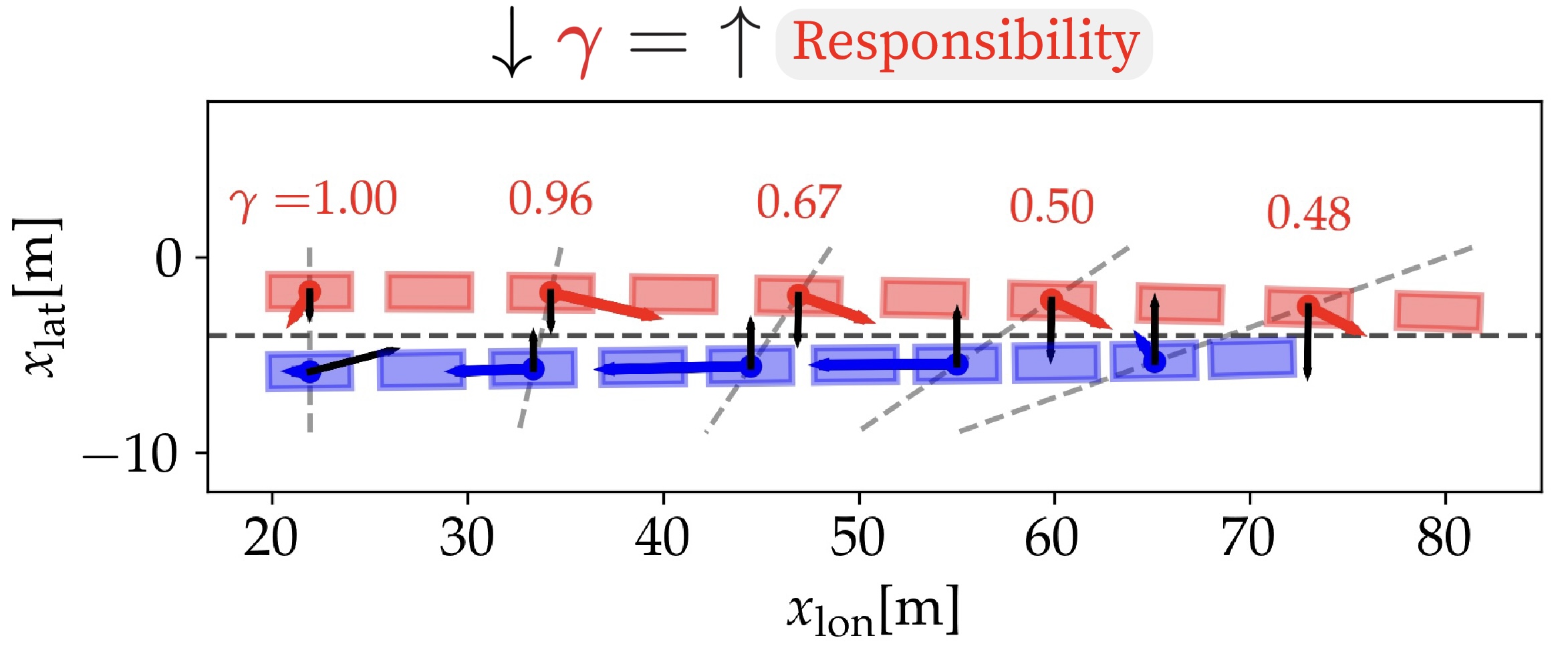}
    \caption{Learned responsibility allocations from a single traffic weaving trajectory. The shown $\gamma$ values correspond to the red agent, and the black and colored arrows represent each agent's desired and actual control, respectively (magnitudes enlarged for visual clarity).}
    \label{fig:gamma-trajectory}
    \vspace{-5mm}
\end{figure}

\section{Experiments on synthetic dataset}
\label{sec:synthetic data}

We evaluate our approach on a synthetically generated dataset to verify that we can recover the ground truth values. We consider a 2-agent system with 1D single integrator dynamics (see Example~\ref{eg: 2-agent 1D single integrator}) and a 6-agent system with 2D double integrator dynamics. We use a CBF that keeps the distance between the (closest) agents at least 1 unit apart.
Note that the double integrator is a relative degree two system, so we employ high order CBFs. As mentioned in Remark~\ref{re:high order cbf}, the high order CBF inequality requires computing additional terms for the CBF inequality constraint, but otherwise, the overall approach remains the same.

For each system, we randomly sample 128 state and desired control pairs and then, for some choice of $\gamma$, we apply the responsibility-aware CBF filter described in Prob.~\ref{prob:responsible multiagent CBF safety filter}, and then add zero-mean, Gaussian noise (variance of $0.1)$ to the solutions.
Then we execute Alg.~\ref{alg:gradient descent} to learn $\gamma$.
As shown in Figs.~\ref{fig:2-agent} and \ref{fig:6-agent}, for a random initial guess, our estimated responsibility allocation $\gamma$ quickly converges to the ground truth value, even if the ground truth is time-varying (see Fig.~\ref{fig:2-agent}). We use a mini-batch size of 8 and a step size of $\delta=0.005$ and $\delta=0.05$ for the 2-agent and 6-agent problems, respectively. 
Additionally, Fig.~\ref{fig:synthetic computation time} presents computation times for calculating the loss and gradient for Prob.~\ref{prob:responsibility bilevel}. We see that our approach has great potential for real-time applications, such as estimating responsibility allocation online, which would only require small batches of data that are streaming in real-time.
The experiments were performed on a MacBook Pro with a M2 processor.

\section{Experiments on traffic-weaving interaction}

We demonstrate the efficacy of our responsibility allocation learning approach on a traffic-weaving interaction dataset, containing trajectories of two human drivers (in a driving simulator) quickly swapping lanes within a short distance, reminiscent of merging at highway on/off ramps \cite{SchmerlingLeungEtAl2018}. An example is shown in Fig. \ref{fig:full-traj}.
This traffic-weaving dataset presents an interesting case study as it contains diverse driving styles from different drivers with no explicit driving rules describing how drivers yield to each other, in contrast to, say, entering a roundabout where cars in the roundabout have the right of way.

\begin{figure*}
    \centering
    \includegraphics[width=\textwidth]{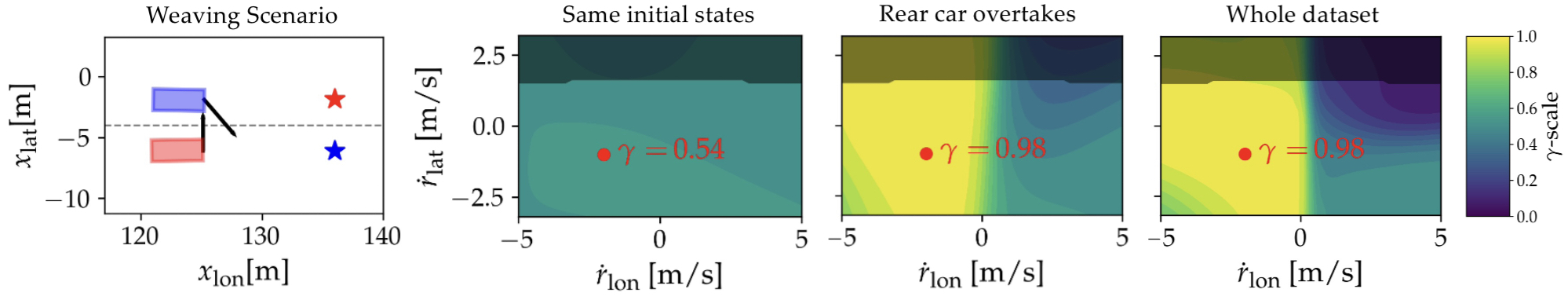}
    \caption{Comparison of $\gamma$ landscapes when learned using different subsets of the traffic weaving dataset (D2, D3, D4; see Sec. \ref{sec:dataset prep} for details). In this scenario, red is moving faster than blue, with the black arrows representing their desired goal-reaching controls. 
    }
    \label{fig:subset-contours}
    \vspace{-5mm}
    
\end{figure*}

\subsection{Experimental set up}

\subsubsection{Traffic-weaving dynamics and CBF construction}
We model each agent with double integrator dynamics, and consider their relative state, $\mathbf{r}=[r_\lon, r_\lat, \dot{r}_\lon, \dot{r}_\lat]=\mathbf{x}_2-\mathbf{x}_1$ where $\mathbf{x}_i=[x_{\lon}, x_{\lat}, \dot{x}_{\lon}, \dot{x}_{\lat}]^T$ is the state of agent $i$, and $\lon$ and $\lat$ denotes the longitudinal and lateral directions respectively. Using relative states reduces the dynamical system's dimensionality.
With these relative double integrator dynamics, we select a CBF $b(\mathbf{r}) = \frac{r_\lon^2}{a_1^2} + \frac{r_\lat^2}{a_2^2} - 1$. We set $a_1=9.22$ and $a_2=1.76$ based on the cars' sizes. Due to the double integrator being relative degree two, we employ a high-order CBF (see Remark \ref{re:high order cbf}).

\subsubsection{Desired control.}
Our proposed approach requires knowledge of each agent's \textit{desired control}.
Unfortunately, this information is not typically known.
In this paper, we use a desired control scheme handcrafted from simple heuristics and empirical observation, leaving exploration of learning agents' desired control policies for future work.

For an agent at state $\mathbf{x}$, its desired lateral control is,
\begin{equation}
    \pi^\desire_\lat(\mathbf{x};\ x_\lat^\desire) = -(x_\lon + \zeta)\beta\tanh(\alpha(x_\lat - x_\lat^\desire)),
    \label{eq:u_des_lateral}
\end{equation}
where $\zeta$, $\beta$ and $\alpha$ are tunable parameters, set as: $\zeta=4.7$, $\beta=0.022$, $\alpha=0.8$, and $x_\lat^\desire$ is the agent's desired lateral position, i.e., the center of the other lane.
Intuitively, the agent's lateral control/steering will get more drastic the longer they take to change lanes.

An agent's desired longitudinal control is,
\begin{equation}
    \pi^\desire_\lon(\mathbf{r} =\mathbf{x}^\prime - \mathbf{x}) =
    \begin{cases} 
          0 & r_\lon > 0 \\
          \frac{-\kappa}{2}(\tanh(r_\lon \dot{r}_\lon) - 1) & r_\lon \leq 0
    \end{cases}
    \label{eq:u_des_long}
\end{equation}
where an agent, if ahead of the other agent, is encouraged to speed up to achieve a larger relative longitudinal distance up to some threshold. Otherwise, if the agent is behind the other agent, it will maintain speed. The parameter $\kappa$ influences the limits of the desired control, which we set to $\kappa=2$.

\subsubsection{Dataset preparation.}
\label{sec:dataset prep}
As mentioned previously, the traffic-weaving dataset consists of diverse behaviors owing to the various driving styles of the driving participants.
As such, we consider four different subsets of the dataset to perform our analysis.
\textbf{(D1)} Single traffic-weaving trajectory.
\textbf{(D2)} Trajectories where both cars start side-by-side at equal initial longitudinal speed.
\textbf{(D3)} Trajectories where one car starting behind but traveling faster overtakes the other car. 
\textbf{(D4)} Whole dataset.

Before training, we consider two types of data augmentation: \textbf{(A1)} mirroring the agents' lateral states and controls across the lane divider ($x_{\mathrm{lat}}, \dot{x}_{\mathrm{lat}}, u_{\mathrm{lat}}, u^{\mathrm{des}}_{\mathrm{lat}} \rightarrow -x_{\mathrm{lat}}, -\dot{x}_{\mathrm{lat}}, -u_{\mathrm{lat}}, -u^{\mathrm{des}}_{\mathrm{lat}}$), and \textbf{(A2)} swapping the states of the agents ($\mathbf{r} \rightarrow -\mathbf{r}$). We apply augmentation A1 to every data subset to remove bias in which lane an agent starts. 
We apply augmentation A2 on dataset D1 to verify the benefits of a symmetric $\gamma$-function; we expand on this in Sec.~\ref{subsec:single-traj-train}. Note: A2 is the ``data augmentation'' referred to in Fig. \ref{fig:symmetry-contours}.

\subsubsection{Model parameterization and training.}
We parameterize the responsibility allocation function as a neural network $\phi$ that takes in the relative state $\mathbf{r}$ and outputs the responsibility allocation. Our $\phi(\mathbf{r})$ function is a small multi-layer perception with three linear hidden layers, hidden layer size of sixteen, and $\tanh$ activation functions, whose final output is passed through the function in (\ref{eq:relative gamma}) to enforce symmetry. Our setup used JAX \cite{JAX2018} and Equinox \cite{KidgerGarcia2021}. We trained on each dataset for $3000$ epochs with batch sizes of $16$ for the single weaving trajectory and $256$ for the other datasets, and used the ADAM optimizer \cite{KingmaBa2015} with a learning rate of $10^{-3}$. Our loss function $\Delta$ from Prob. \ref{prob:responsibility bilevel} is the Huber loss, and the regularization/slack coefficients $\beta_1$ and $\beta_2$ from Prob. \ref{prob:CBF safety filter} are tuned to $0.1$ and $600$ respectively.

\subsection{Training on a single trajectory (D1)}
\label{subsec:single-traj-train}
We first apply our learning method to the D1 dataset for a sanity check and validate our intuition.
In this single traffic-weaving trajectory (similar to Fig. \ref{fig:full-traj}), the car in the upper lane starts $2$m/s faster and $3$m behind the lower car, before overtaking the lower car and switching lanes first. 

\subsubsection{Benefits of using a symmetric model}
Fig.~\ref{fig:symmetry-contours} shows the learned responsibility allocation for the red car ${\color{red}\gamma_\mathrm{red}}$ (and the blue car's responsibility allocation is ${\color{blue}\gamma_\mathrm{blue}}=1 - {\color{red}\gamma_\mathrm{red}}$). The top and bottom rows are in the same scenario except for the cars swap places. This means we expect ${\color{red}\gamma_\mathrm{red}^\text{top scene}} = {\color{blue}\gamma_\mathrm{blue}^\text{bottom scene}}.$ This is precisely the symmetric responsibility property, and we see that, as expected, a symmetric model satisfies this property exactly whereas the unconstrained model does not.

Additionally, note that the results with the symmetric model require no augmenting of swapped players' states, and we see that it achieves very similar results to the unconstrained model using this data augmentation. The unconstrained model without data augmentation has poor results when the cars are in a configuration unseen in the original trajectory.
These results highlight the data efficiency benefits of using a symmetric model---data augmentation may be unnecessary, a benefit when working with more complex scenarios where data augmentation is less trivial.
Moving forward, we use a symmetric responsibility allocation model.

\subsubsection{Interpreting $\gamma$ over a trajectory}
Fig. \ref{fig:gamma-trajectory} shows how the value of $\gamma$ changes during the trajectory as the two cars swap lanes. The evolution of $\gamma$ is consistent with the alignment between the agents' actual (colored) control and desired (black) control vectors; initially, the red car has low responsibility as it does not yet need to change lanes. However, after overtaking the blue car, its responsibility increases because it becomes more eager to swap lanes but cannot because of the blue car. The blue car, meanwhile, has to brake early in the trajectory when it wants to move laterally (hence the greater responsibility early on), as it must accommodate the red car overtaking it. This example illustrates how  we can learn intuitive responsibility allocations for a basic traffic weaving trajectory.

\subsection{Training on data with same initial conditions (D2)} 
We apply our learning method to the D2 dataset to see what happens when the dataset contains diverse and contradictory behaviors.
As the initial conditions place the cars on equal footing at the start, there is ambiguity as to who passes whom, resulting in multimodal behavior within the dataset.
Learning responsibility allocations in this setting is substantially more difficult; as shown by Fig. \ref{fig:subset-contours}\textcolor{red}: \textit{Same initial states}. The learned responsibility allocation is constant---regardless of their relative state, both agents are predicted to be equally responsible. This performance drop under multimodal interactions highlights a limitation in our approach, necessitating future work on a \textit{probabilistic} extension.

\subsection{Training on data when rear, faster car overtakes (D3)}
We apply our learning method to the D3 dataset to investigate if it can successfully capture the trend that the slower car yields to the faster car.
In Fig. \ref{fig:subset-contours}\textcolor{red}: \textit{Rear car overtakes}, we see that with the same longitudinal position, the faster car has a much larger $\gamma$ value, meaning they are less responsible for maintaining safety, consistent with behavior in the dataset where the slower car yields to the faster car.

\subsection{Training on the whole dataset (D4)}
Finally, we apply our learning methods to the entire traffic-weaving dataset (D4).
Fig. \ref{fig:subset-contours}\textcolor{red}: \textit{Whole dataset} shows that it learned a responsibility allocation model similar to the model trained on D3. Despite D4 containing more diverse trajectories, there is a bias towards cars starting behind but traveling faster to overtake rather than yield to the slower but front car \cite{SchmerlingLeungEtAl2018}.
Clearly, our model captures this bias.

\subsection{Discussion}
From the results of our different training scenarios on the traffic weaving dataset, we draw three key takeaways:
\begin{enumerate}
    \item Our responsibility model and inference technique was able to efficiently perform on real human trajectory data, assigning interpretable quantities to observed behaviors. This affirms that we can transfer our definition of responsibility allocation from theory to practice. 
    \item Our approach learned nontrivial responsibility allocations on data with consistent bias (e.g., a significant portion of data having the faster car overtake the slower car), but struggled with data containing multimodal behaviors. This suggests a need to investigate probabilistic extensions to capture multimodality.
    \item Identifying the ``best'' desired control policy is non-trivial. While our handcrafted nonlinear control policy for the traffic weaving data lets us learn intuitive responsibility allocations, how we can do better is unclear; investigating principled approaches to learning desired policies from data is a logical next step.
\end{enumerate}

\section{Conclusions and future work}
In this paper, we introduced a novel responsibility allocation framework that distills ill-defined social norms underpinning human interactions into an interpretable quantity. 
We presented a mathematical formulation of responsibility rooted in Control Barrier Functions and a computationally efficient approach to learn responsibility allocations from data via differentiable optimization.
We also proposed an approach to learning symmetric responsibility allocations, improving data efficiency. 
We showed the efficacy of our algorithm on synthetic and human interaction data, and results indicate we can gain interpretable insights into the social dynamics of multi-agent interactions.
For future work, we will \textbf{(i)} investigate principled ways to construct desired control policies, e.g. via learned generative models, \textbf{(ii)} develop a probabilistic extension to this framework to account for multimodal interactions, and \textbf{(iii)} examine using responsibility allocations for guiding robot policy construction.

\section*{Acknowledgment}
We thank Hamzah Khan, Xinjie Liu, Dong Ho Lee, Griffin Norris, and Kazuki Mizuta for their feedback and discussion.

\bibliographystyle{IEEEtran}
\bibliography{main,ctrl_papers}

\end{document}

%% file: notation.tex
\newcommand{\safe}{\mathrm{safe}}
\newcommand{\desire}{\mathrm{des}}

\newcommand{\execute}{\mathrm{exec}}
\newcommand{\data}{\mathrm{data}}
\newcommand{\gammavec}{\boldsymbol{\gamma}}

\newcommand{\lon}{\mathrm{lon}}
\newcommand{\lat}{\mathrm{lat}}

%% file: ctrl_papers.bib
@Preamble{"\newcommand{\noopsort}[1]{} " #
"\newcommand{\printfirst}[2]{#1} " #
"\newcommand{\singleletter}[1]{#1} " #
"\newcommand{\switchargs}[2]{#2#1} "}

@String{jrn_IEEE_RAL                   = {{IEEE Robotics and Automation Letters}}}

@String{jrn_IEEE_TAC                   = {{IEEE Transactions on Automatic Control}}}

@String{jrn_IEEE_TIV                   = {{IEEE Transactions on Intelligent Vehicles}}}

@String{jrn_PNAS                       = {{Proceedings of the National Academy of Sciences}}}

@String{proc_ICLR                      = {{Int.\ Conf.\ on Learning Representations}}}

@String{proc_ICML                      = {{Int.\ Conf.\ on Machine Learning}}}

@String{proc_IEEE_CDC                  = {{Proc.\ IEEE Conf.\ on Decision and Control}}}

@String{proc_IEEE_ICRA                 = {{Proc.\ IEEE Conf.\ on Robotics and Automation}}}

@String{proc_IEEE_IROS                 = {{IEEE/RSJ Int.\ Conf.\ on Intelligent Robots \& Systems}}}

@String{proc_L4DC                      = {{Learning for Dynamics \& Control Conference}}}

@String{proc_NIPS                      = {{Conf.\ on Neural Information Processing Systems}}}

@String{proc_NIPS-AD                   = {{Conf.\ on Neural Information Processing Systems - Autodiff Workshop}}}

@String{proc_RSS                       = {{Robotics: Science and Systems}}}

@InProceedings{SchmerlingLeungEtAl2018,
  author    = {Schmerling, E. and Leung, K. and Vollprecht, W. and Pavone, M.},
  booktitle = proc_IEEE_ICRA,
  title     = {Multimodal Probabilistic Model-Based Planning for Human-Robot Interaction},
  year      = {2018},
  arxiv     = {1710.09483},
  category  = {interaction},
  img       = {SchmerlingLeungEtAl2018.png},
  selected  = {true},
}

@Article{IvanovicLeungEtAl2020,
  author   = {Ivanovic, B.* and Leung, K.* and Schmerling, E. and Pavone, M.},
  journal  = jrn_IEEE_RAL,
  number   = {2},
  pages    = {295--302},
  title    = {Multimodal Deep Generative Models for Trajectory Prediction: A Conditional Variational Autoencoder Approach},
  volume   = {6},
  arxiv    = {2008.03880},
  category = {interaction},
  img      = {IvanovicLeungEtAl2020.jpg},
  selected = {true},
  year     = {2021},
}

@InProceedings{CosnerChenEtAl2023,
  author    = {Cosner, R. and Chen, Y. and Leung, K. and Pavone, M.},
  booktitle = proc_IEEE_ICRA,
  title     = {Learning Responsibility Allocations for Safe Human-Robot Interaction with Applications to Autonomous Driving},
  year      = {2023},
  arxiv     = {2303.03504},
  category  = {structure},
  img       = {CosnerChenEtAl2023.png},
  selected  = {true},
}

@InProceedings{MizutaLeung2024,
  author    = {Mizuta, K. and Leung, K.},
  booktitle = proc_IEEE_IROS,
  title     = {{CoBL-Diffusion}: Diffusion-Based Conditional Robot Planning in Dynamic Environments Using Control Barrier and Lyapunov Functions},
  year      = {2024},
  arxiv     = {2406.05309},
  img       = {MizutaLeung2024.png},
}


%% file: main.bib
@String{jrn_ACM_HRI                    = {{ACM Transactions on Human-Robot Interaction}}}

@String{jrn_IEEE_JAS                   = {{IEEE/CAA Journal of Automatica Sinica}}}

@String{proc_ACM_ICAICS                = {{ACM Int.\ Conf.\ on Artificial Intelligence and Computer Science}}}

@Inproceedings{SadighSastryEtAl2016c,
  author    = {Sadigh, D. and Sastry, S. and Seshia, S. A. and Dragan, A. D.},
  title     = {Planning for Autonomous Cars that Leverage Effects on Human Actions},
  booktitle = proc_RSS,
  year      = {2016},
}

@article{StellatoBanjacEtAl2017,
  author    = {Stellato, B. and Banjac, G. and Goulart, P. and Bemporad, A. and Boyd, S.},
  title     = {{OSQP}: An Operator Splitting Solver for Quadratic Programs},
  year      = {2017},
  journal      = {{Available at } \url{https://arxiv.org/abs/1711.08013}},
}

@InProceedings{PaszkeGrossEtAl2017,
  author    = {Paszke, A. and Gross, S. and Chintala, S. and Chanan, G. and Yang, E. and DeVito, Z. and Lin, Z. and Desmaison, A. and Antiga, L. and Lerer, A.},
  booktitle = proc_NIPS-AD,
  title     = {Automatic differentiation in {PyTorch}},
  year      = {2017},
}

@Article{RevelsLubinEtAl2016,
  author    = {Revels, J. and Lubin, M. and Papamarkou, T.},
  journal   = {{Available at }\url{https://arxiv.org/abs/1607.07892}},
  title     = {Forward-Mode Automatic Differentiation in Julia},
  year      = {2016},
}

@Article{GuoWangEtAl2021,
  author  = {Guo, Ke and Wang, Dawei and Fan, Tingxiang and Pan, Jia},
  journal = jrn_IEEE_RAL,
  number  = {3},
  pages   = {4520--4527},
  title   = {{VR-ORCA}: Variable Responsibility Optimal Reciprocal Collision Avoidance},
  volume  = {6},
  year    = {2021},
}

@Article{JAX2018,
  author  = {Bradbury, J. and Frostig, R. and Hawkins, P. and Johnson, M.~J. and Leary, C. and Maclaurin, D. and Necula, G. and Paszke, A. and Vander{P}las, J. and Wanderman-{M}ilne, S. and Zhang, Q.},
  journal = {Available at http://github.com/google/jax},
  title   = {{JAX}: composable transformations of {P}ython+{N}um{P}y programs},
  year    = {2018},
}

@Article{XiaoBelta2021,
  author  = {Xiao, W. and Belta, C.},
  journal = jrn_IEEE_TAC,
  number  = {7},
  pages   = {3655--3662},
  title   = {High Order Control Barrier Functions},
  volume  = {67},
  year    = {2021},
}

@InProceedings{VandenBergLinEtAl2008,
  author    = {van den Berg, J. and Lin, M. and Manocha, D.},
  booktitle = proc_IEEE_ICRA,
  title     = {Reciprocal velocity obstacles for real-time multi-agent navigation},
  year      = {2008},
}

@Article{SchwartingPiersonEtAl2019,
  author  = {Schwarting, W. and Pierson, A. and Alonso-Mora, J. and Karaman, S. and Rus, D.},
  journal = jrn_PNAS,
  number  = {50},
  pages   = {24972--24978},
  title   = {Social behavior for autonomous vehicles},
  volume  = {116},
  year    = {2019},
}

@InProceedings{SunZhanEtAl2018,
  author    = {Sun, L. and Zhan, W. and Tomizuka, M. and Dragan, A.},
  booktitle = proc_IEEE_IROS,
  title     = {Courteous Autonomous Cars},
  year      = {2018},
}

@InProceedings{AgrawalAmosEtAl2019,
  author={Agrawal, A. and Amos, B. and Barratt, S. and Boyd, S. and Diamond, S. and Kolter, Z.},
  title={Differentiable Convex Optimization Layers},
  booktitle=proc_NIPS,
  year={2019},
}

@Article{MavrogiannisAlves-OliveraEtAl2021,
  author  = {Mavrogiannis, C. I. and Alves-Olivera, P. and Thomason, W. and Knepper, R. A.},
  journal = jrn_ACM_HRI,
  number  = {4},
  title   = {Social Momentum: Design and Evaluation of a Framework for Socially Competent Robot Navigation},
  volume  = {37},
  year    = {2021},
}

@InProceedings{ToghiValienteEtAl2021,
  author    = {Toghi, B. and Valiente, R. and Sadigh, D. and Pedarsani, R. and Fallah, Y. P.},
  booktitle = proc_IEEE_IROS,
  title     = {Cooperative Autonomous Vehicles that Sympathize with Human Drivers},
  year      = {2021},
}

@Article{ToghiValienteEtAl2022,
  author  = {Toghi, B. and Valiente, R. and Sadigh, D. and Pedarsani, R. and Fallah, Y.},
  journal = jrn_IEEE_TIV,
  number  = {12},
  pages   = {24791--24804},
  title   = {Social Coordination and Altruism in Autonomous Driving},
  volume  = {23},
  year    = {2022},
}

@InProceedings{NayakantiAlRfouEtAl2023,
  author    = {Nayakanti, N. and Al-Rfou, R. and Zhou, A. and Goel, K. and Refaat, K.~S. and Sapp, B.},
  booktitle = proc_IEEE_ICRA,
  title     = {Wayformer: Motion Forecasting via Simple and Efficient Attention Networks},
  year      = {2023},
}

@Article{AgrawalBarrattEtAl2021,
  author  = {Agrawal, A. and Barratt, S. and Boyd, S.},
  journal = jrn_IEEE_JAS,
  number  = {8},
  pages   = {1355--1364},
  title   = {Learning Convex Optimization Models},
  volume  = {8},
  year    = {2021},
}

@InProceedings{AmosKolter2017,
  author    = {Amos, B. and Kolter, J.~Z.},
  booktitle = proc_ICML,
  title     = {Optnet: Differentiable optimization as a layer in neural networks},
  year      = {2017},
}

@InProceedings{AgrawalBarrattEtAl2020,
  author    = {Agrawal, A. and Barratt, S. and Boyd, S. and Stellato, B.},
  booktitle = proc_L4DC,
  title     = {Learning Convex Optimization Control Policies},
  year      = {2020},
}

@InProceedings{AmesGrizzleEtAl2014,
  author    = {Ames, A.~D. and Grizzle, J.~W. and Tabuada, P.},
  booktitle = proc_IEEE_CDC,
  title     = {Control barrier function based quadratic programs with application to adaptive cruise control},
  year      = {2014},
}

@Article{AmesXuEtAl2017,
  author  = {Ames, A.~D. and Xu, X. and Grizzle, J.~W. and Tabuada, P.},
  journal = jrn_IEEE_TAC,
  number  = {8},
  pages   = {3861--3876},
  title   = {Control Barrier Function Based Quadratic Programs for Safety Critical Systems},
  volume  = {62},
  year    = {2017},
}

@InProceedings{KingmaBa2015,
  author    = {Kingma, D. and Ba, J.},
  booktitle = proc_ICLR,
  title     = {Adam: A Method for Stochastic Optimization},
  year      = {2015},
}

@Article{TracyManchester2024,
  author  = {Tracy, K. and Manchester, Z.},
  journal = {{Available at } \url{https://arxiv.org/abs/2406.11749}},
  title   = {On the Differentiability of the Primal-Dual Interior-Point Method},
  year    = {2024},
}

@Article{LuoCaiEtAl2018,
  author  = {Luo, Y. and Cai, P. and Bera, A. and Hsu, D. and Lee, W.~S. and Manocha, D.},
  journal = jrn_IEEE_RAL,
  number  = {4},
  pages   = {3418--3425},
  title   = {{PORCA}: Modeling and planning for autonomous driving among many pedestrians},
  volume  = {3},
  year    = {2018},
}

@Article{LuoCaiEtAl2022,
  author  = {Luo, Y. and Cai, P. Lee, Y. and Hsu, D.},
  journal = jrn_IEEE_RAL,
  number  = {2},
  pages   = {3499--3506},
  title   = {{GAMMA}: A general agent motionprediction model for autonomous driving},
  volume  = {7},
  year    = {2022},
}

@InProceedings{YinLiuEtAl2019,
  author    = {Yin, Z. and Liu, J. and Wang, L.},
  booktitle = proc_ACM_ICAICS,
  title     = {Less-effort collision avoidance in virtual pedestrian simulation},
  year      = {2019},
}

@InProceedings{KidgerGarcia2021,
  author    = {Kidger, P. and Garcia, C.},
  booktitle = proc_NIPS,
  title     = {{E}quinox: neural networks in {JAX} via callable {P}y{T}rees and filtered transformations},
  year      = {2021},
}
